\documentclass[12pt]{article}
\usepackage{geometry}

\usepackage[utf8]{inputenc}
\usepackage{amssymb}
\usepackage{amsmath}
\usepackage[english]{babel}

\usepackage{amsthm}
\newtheorem{theorem}{Theorem}[section]
\newtheorem{remark}{Remark}[section]
\newtheorem{proposition}{Proposition}[section]
\newtheorem{Corollary}{Corollary}[section]

\usepackage{color}
\usepackage{mathrsfs}
\usepackage{graphicx}

\allowdisplaybreaks

\usepackage{hyperref}
\numberwithin{equation}{subsection}
\hypersetup{
  pdftitle={Superintegrable metrics on surfaces admitting integrals of degrees 1 and 4},
  pdfauthor={Pavel Novichkov},
  bookmarksnumbered=true,
  bookmarksopen=false,
  bookmarksopenlevel=1,
  colorlinks=true,
  linkcolor=blue,
  citecolor=blue,
  urlcolor=blue,
  pdfpagemode=UseOutlines
}
\usepackage{cleveref}
\crefname{equation}{}{}
\crefrangelabelformat{equation}{(#3#1#4--#5#2#6)}
\crefmultiformat{equation}{(#2#1#3}{, #2#1#3)}{#2#1#3}{#2#1#3}

\newcommand{\N}{\mathbb{N}}
\newcommand{\Z}{\mathbb{Z}}
\newcommand{\R}{\mathbb{R}}
\newcommand{\C}{\mathbb{C}}

\newcommand{\J}{\mathscr{J}_4}  
\newcommand{\Mgen}{\mathcal{M}}
\newcommand{\M}{\mathcal{M}^2}  
\renewcommand{\L}{\mathcal{L}}  
\newcommand{\EEE}{\mathcal{E}^{(3)}}  
\newcommand{\EE}{\mathcal{E}^{(2)}}  
\newcommand{\Ef}{\mathcal{E}^{(1)}_1}  
\newcommand{\Es}{\mathcal{E}^{(1)}_2}  
\newcommand{\D}{\mathscr{D}}  
\newcommand{\I}{\mathcal{I}}  

\renewcommand{\d}{\partial}  
\renewcommand{\i}{\mathrm{\bf{i}}\,}  
\newcommand{\bz}{\bar{z}}  
\newcommand{\pbz}{p_{\bz}}  
\newcommand{\pxi}{p_{\xi}}
\newcommand{\peta}{p_{\eta}}

\DeclareMathOperator{\Span}{Span}

\begin{document}


\thispagestyle{empty}

\begin{center}
{\LARGE \textsf{\textbf{
  Superintegrable metrics on surfaces admitting integrals of degrees 1 and 4}}
}
\end{center}

\vspace{0.1cm}

\begin{center}
{\large Pavel Novichkov}
\end{center}

\begin{center}
\textsl{
  National Research University Higher School of Economics,\\
  Faculty of Mathematics, Moscow, Russia
}
\end{center}

\begin{center}
Master's thesis\\
Supervisor: Vsevolod Shevchishin, Dr.rer.nat., Dr.habil.\\
\vspace{0.2cm}
June 2015
\end{center}

\vspace{0.4cm}
\hrule
\vspace{0.3cm}
\noindent{\sffamily \bfseries Abstract} \\[0.1cm]
We study Riemannian metrics on~surfaces whose geodesic flows are superintegrable with one integral linear in~momenta and one integral quartic in~momenta.
The main results of~the~work are~local description of~such metrics in~terms of~ordinary differential equations, integration of the~equations, and~description of~the corresponding Poisson algebra of~integrals of~motion.
We also give examples of~such metrics that can be extended to the~sphere~$S^2$, and study the~group of~symmetries of~the problem.
\vskip 10pt
\hrule

\tableofcontents

\newpage

\section{Introduction}

Let~$\Mgen$~be a~smooth $n$-dimensional Riemannian manifold with metric~$g = (g_{ij})$.
Let~$g^* = (g^{ij})$ be a dual metric tensor which defines a~metric on the~cotangent bundle~$T^*\Mgen$.
It~is well known that the~geodesic flow on~$\Mgen$ can be described as the~Hamiltonian flow on~$T^*\Mgen$ with the~Hamiltonian~$H_g := \frac12 g^*$.
The~corresponding mechanical system on~$T^*\Mgen$ is called {\it free}.
$H_g$ is called its {\it kinetic energy}.
Recall also that a~Hamiltonian mechanical system with Hamiltonian~$H$ defined on a~symplectic manifold~$(X,\omega)$ is called {\it (Liouville) integrable} if it allows~$n$ functionally independent integrals~$H = X_1,\, X_2,\ldots,X_n$ in~involution.
An~integrable system is called {\it superintegrable} if it allows further functionally independent integrals~$Y_1,\, Y_2,\ldots,Y_k,\, 1 \leq k \leq n - 1$ (not necessarily in involution with each other nor with~$X_i$).
If a~free Hamiltonian system is superintegrable then the~corresponding metric~$g$ is also called {\it superintegrable}.

Classical examples of~integrable systems include Kepler-Coulomb problem, 2 and 3-dimensional isotropic harmonic oscillator, and hydrogen atom.
All of them turn out to~be superintegrable.
For~instance, Kepler-Coulomb problem has 7 integrals of motion (energy, 3 components of angular momentum vector and 3 components of Laplace-Runge-Lenz vector).
These integrals are related by 2 scalar equations, hence there are 5 functionally independent integrals.

Superintegrable systems are of great interest for modern mathematics and mathematical physics due to several special properties they have:

\begin{itemize}
\item
  if a~superintegrable system has $k$ extra integrals, then its trajectories lie on $(n - k)$-dimensional submanifolds;
\item
  in~particular, when~$k = n - 1$ (so called {\it maximal superintegrability})
  all trajectories lie on 1-dimensional submanifolds, hence all finite trajectories are closed;
\item
  integrals of motion generate a~nontrivial Poisson algebra;
\item
  extra integrals correspond to {\it hidden symmetries};
\item
  in~quantum physics, superintegrability leads to {\it accidental degeneracy} of energy levels
\end{itemize}
etc.~\cite{Winternitz}.

One natural question that arises here is how to~find and describe all superintegrable systems in simple cases.
The~simplest nontrivial possibility in~terms of dimension is~$n = 2$.
A~2-dimensional superintegrable system has 2 extra integrals besides energy.
Obviously, such a~system is maximally superintegrable.
Let us also restrict ourselves to considering free superintegrable systems (i.e., searching for superintegrable metrics).
Finally, let us assume that integrals of~motion are polynomial in~momenta.
This assumption is motivated by real-world physical examples of such systems, and by the~Whittaker theorem (if a~system allows an~integral analytical in~momenta then it also allows an~integral polynomial in~momenta~\cite{Whittaker}).

One can easily show that each homogeneous component of~a~polynomial integral is an~integral by~itself, therefore it is sufficient to consider homogeneous integrals of~degrees~$l$ and~$m$.
Let us call this case~``$l+m$''.

In~the case~``$1+1$'' each superintegrable~metric has constant curvature.
Cases~``$1+2$'' and~``$2+2$'' were described locally in~the~classical work by Koenigs~\cite{Koenigs}.
Finally, local description of~``$1+3$''~metrics was obtained in~\cite{Matveev_Shevchishin,Valent_Duval_Shevchishin}.

The~aim of~this work is to~describe superintegrable metrics in~the next case~``$1+4$''.
We develop methods and ideas used in~\cite{Matveev_Shevchishin} to~solve the~case~``$1+3$''.
In Chapter~\ref{ch:local_description}, we formulate and prove the~theorem which gives local description of ``$1+4$'' superintegrable metrics in~the~neighborhood of almost every point.
In Chapter~\ref{ch:further}, we study symmetries of such systems and obtain algebraic forms of the~equations describing the metrics.
It allows us to construct examples of the metrics that can be extended to the~sphere.

The~author is grateful to his advisor Vsevolod Shevchishin for
his continuous support and many useful discussions.

\section{Local description of 2-dimensional superintegrable metrics admitting integrals of degrees 1 and 4}
\label{ch:local_description}

\subsection{Main result}

\begin{theorem}
  \label{th:local}
  Let $\M$ be a~smooth 2-dimensional manifold equipped with a~Riemannian metric~$g = (g_{ij})$.
  Let $H = \frac12 g^{ij}p_ip_j$ be the~corresponding Hamiltonian.

  Suppose the~geodesic flow of~$g$ admits integrals~$L$~and~$F$ of degrees 1 and 4 such that~$L,\, F$ and~$H$ are functionally independent.
  Then in~the~neighborhood of almost every point~$X \in \M$ (such that $L|_{T^*_X \M} \not\equiv 0$) there exist local coordinates~$(x,y)$ such that the~metric has the~form~$g = \frac{1}{h_x^2} (dx^2 + dy^2)$ where $h_x \equiv \dfrac{dh}{dx}$ and~$h(x)$ is a~real function that satisfies one of the~following systems of~equations:
  \begin{subequations}
    \label{eq:main}
    \begin{align}
      \label{eq:trig}
      &
      \begin{cases}
        \begin{gathered}
          2 \left(
            A_0 h - A_1
          \right)
          \left(
            A_3 \cos(\mu x) + \frac{A_4}{\mu} \sin(\mu x)
          \right) h_x^3\\
          + \left[
            \left(
              \mu \left(
                A_0 h - 2 A_1
              \right) h + 2\, \frac{A_2}{\mu}
            \right)
            \left(
              A_3 \sin(\mu x) - \frac{A_4}{\mu} \cos(\mu x)
            \right) + A_5
          \right] h_x^2\\
          - \left(
            A_3 \cos(\mu x) + \frac{A_4}{\mu} \sin(\mu x)
          \right)^2 = 0,
        \end{gathered}\\
        \begin{gathered}
          \left(
            A_0 h - A_1
          \right)^2 h_x^3 \,
          + \left[
            \frac{\mu}{2} \left(
              A_0 h - 2 A_1
            \right) h
            \left(
              \frac{\mu}{2} \left(
                A_0 h - 2 A_1
              \right) h + 2\, \frac{A_2}{\mu}
            \right)
          \right.\\
          \left.
            +\, 2\, \frac{A_0}{\mu} \left(
              A_3 \sin(\mu x) - \frac{A_4}{\mu} \cos(\mu x)
            \right) + \frac{A_6}{4 \mu^2}
          \right] h_x\\
          -\, 2 \left(
            A_0 h - A_1
          \right)
          \left(
            A_3 \cos(\mu x) + \frac{A_4}{\mu} \sin(\mu x)
          \right) = 0,
        \end{gathered}
      \end{cases}\\
      \label{eq:hyper}
      &
      \begin{cases}
        \begin{gathered}
          2 \left(
            A_0 h - A_1
          \right)
          \left(
            A_3 \cosh(\mu x) + \frac{A_4}{\mu} \sinh(\mu x)
          \right) h_x^3\\
          - \left[
            \left(
              \mu \left(
                A_0 h - 2 A_1
              \right) h - 2\, \frac{A_2}{\mu}
            \right)
            \left(
              A_3 \sinh(\mu x) + \frac{A_4}{\mu} \cosh(\mu x)
            \right) - A_5
          \right] h_x^2\\
          - \left(
            A_3 \cosh(\mu x) + \frac{A_4}{\mu} \sinh(\mu x)
          \right)^2 = 0,
        \end{gathered}\\
        \begin{gathered}
          \left(
            A_0 h - A_1
          \right)^2 h_x^3 \,
          - \left[
            \frac{\mu}{2} \left(
              A_0 h - 2 A_1
            \right) h
            \left(
              \frac{\mu}{2} \left(
                A_0 h - 2 A_1
              \right) h - 2\, \frac{A_2}{\mu}
            \right)
          \right.\\
          \left.
            -\, 2\, \frac{A_0}{\mu} \left(
              A_3 \sinh(\mu x) + \frac{A_4}{\mu} \cosh(\mu x)
            \right) + \frac{A_6}{4 \mu^2}
          \right] h_x\\
          -\, 2 \left(
            A_0 h - A_1
          \right)
          \left(
            A_3 \cosh(\mu x) + \frac{A_4}{\mu} \sinh(\mu x)
          \right) = 0,
        \end{gathered}
      \end{cases}\\
      \label{eq:linear}
      &
      \begin{cases}
        \begin{gathered}
          2 \left(
            A_0 h - A_1
          \right)
          \left(
            A_3 + A_4 x
          \right) h_x^3
          + \left(
            A_4 h \left(
              2 A_1 - A_0 h
            \right)
          \right.\\
          \left.
            + A_2 x \left(
              2 A_3 + A_4 x
            \right) + A_5
          \right) h_x^2
          - \left(
            A_3 + A_4 x
          \right)^2 = 0,
        \end{gathered}\\
        \begin{gathered}
          \left(
            A_0 h - A_1
          \right)^2 h_x^3
          + \left(
            \left(
              A_0 h - 2 A_1
            \right) A_2 h
            + A_0 x \left(
              2 A_3 + A_4 x
            \right) + A_6
          \right) h_x\\
          -\, 2 \left(
            A_0 h - A_1
          \right)
          \left(
            A_3 + A_4 x
          \right) = 0
        \end{gathered}
      \end{cases}
    \end{align}
  \end{subequations}
  Here $\mu > 0, \, A_0, \ldots, A_6$ are arbitrary real constants.

  Further, the~linear integral~$L$ is given by~$L = p_y$.

  Besides, in the~case~(a) the~metric admits 2-parametric family of quartic integrals of the~form
  \begin{equation}
    F = C_+ F_+ + C_- F_-,
  \end{equation}
  where
  \begin{equation}
    \begin{gathered}
      F_+ = e^{\mu y} \sum_{k=0}^4 a_k(x) p_x^{4-k} p_y^k,\\
      F_- = e^{-\mu y} \sum_{k=0}^4 (-1)^k a_k(x) p_x^{4-k} p_y^k.
    \end{gathered}
  \end{equation}
  Here $C_+,\, C_-$ are arbitrary real constants, and functions $a_k(x)$ are given by
  \begin{equation}
    \begin{aligned}
      & a_0(x) = \frac{A_0}{\mu} h_x^4,\\
      & a_1(x) = (A_1 - A_0 h) h_x^3,\\
      & a_2(x) =
      \left(
        \frac{1}{2} \mu A_0 h^2 + \frac{2A_0}{\mu} h_x^2 - \mu A_1 h + \frac{A_2}{\mu}
      \right) h_x^2,\\
      & a_3(x) = (A_1 - A_0 h) h_x^3 - \left(
        A_3 \cos (\mu x) + \frac{A_4}{\mu} \sin (\mu x)
      \right),\\
      & \begin{gathered}
        a_4(x) = \left(
          \frac{\mu}{2} (A_0 h - 2 A_1) h + \frac{A_0}{\mu} h_x^2 + \frac{A_2}{\mu}
        \right) h_x^2
        + \left(
          A_3 \sin (\mu x) - \frac{A_4}{\mu} \cos (\mu x)
        \right).
      \end{gathered}
    \end{aligned}
  \end{equation}

  Next, in the~case~(b) the~metric admits 2-parametric family of quartic integrals of the~form
  \begin{equation}
    \begin{aligned}
      F = &\,\, C_e \cos(\mu y + \phi_0)\, \left(a_0(x)\, p_x^4 + a_2(x)\, p_x^2 p_y^2 + a_4(x)\, p_y^4\right)\\
      + &\,\, C_e \sin(\mu y + \phi_0)\, \left(a_1(x)\, p_x^3 p_y + a_3(x)\, p_x p_y^3\right),
    \end{aligned}
  \end{equation}
  where~$C_e,\, \phi_0$ are arbitrary real constants, and functions~$a_k(x)$ are given by
  \begin{equation}
    \begin{aligned}
      & a_0(x) = - \frac{A_0}{\mu} h_x^4,\\
      & a_1(x) = (A_1 - A_0 h) h_x^3,\\
      & a_2(x) =
      \left(
        \frac{1}{2} \mu A_0 h^2 - \frac{2A_0}{\mu} h_x^2 - \mu A_1 h - \frac{A_2}{\mu}
      \right) h_x^2,\\
      & a_3(x) = (A_1 - A_0 h) h_x^3 - \left(
        A_3 \cosh (\mu x) + \frac{A_4}{\mu} \sinh (\mu x)
      \right),\\
      & \begin{gathered}
        a_4(x) = \left(
          \frac{\mu}{2} (A_0 h - 2 A_1) h - \frac{A_0}{\mu} h_x^2 - \frac{A_2}{\mu}
        \right) h_x^2
        + \left(
          A_3 \sinh (\mu x) + \frac{A_4}{\mu} \cosh (\mu x)
        \right).
      \end{gathered}
    \end{aligned}
  \end{equation}

  Finally, in the~case~(c) the~metric admits a~quartic integral of the~form
  \begin{equation}
    \begin{gathered}
      F = y \left(
        a_0(x) p_x^4 + a_2(x) p_x^2 p_y^2 + a_4(x) p_y^4
      \right) +
      a_1(x) p_x^3 p_y + a_3(x) p_x p_y^3\,,
    \end{gathered}
  \end{equation}
  where functions~$a_k(x)$ are given by
  \begin{equation}
    \begin{aligned}
      & a_0(x) = A_0 h_x^4,\\
      & a_1(x) = (A_1 - A_0 h) h_x^3,\\
      & a_2(x) = (2 A_0 h_x^2 + A_2) h_x^2,\\
      & a_3(x) = (A_1 - A_0 h) h_x^3 - (A_3 + A_4 x),\\
      & a_4(x) = -A_4 + A_2 h_x^2 + A_0 h_x^4.
    \end{aligned}
  \end{equation}
\end{theorem}

\begin{remark}
  Actually, any equation in each pair~\eqref{eq:main} is sufficient to describe the~function~$h(x)$.
  However, second equation allows to~eliminate~$h_x$ (see~\ref{ch:algebraic} for details), therefore the~system of two equations is an~implicit way to~integrate each of them.
  The~constants~$A_5$ and~$A_6$ play an~important role here: in~fact, the~choice of~$A_6$ for the~second equation is equivalent to the~choice of initial conditions for the~first equation, and vice versa.
\end{remark}

\subsection{Reduction to a 3rd order ODE}

Since the~metric~$g$ allows the~linear integral~$L$ it follows that in the~neighborhood of every point $X \in \M$ such that $L|_{T^*_X \M} \not\equiv 0$ there exist local coordinates~$(x,y)$ such that the~metric and the~linear integral have the~form
\begin{equation}
  g = \lambda (x) \left( dx^2 + dy^2 \right), \quad L=p_y,
\end{equation}
where~$\lambda (x)$ is a~real function \cite{Bolsinov_Matveev_Fomenko}.
It means that superintegrability condition can be expressed as a~condition on the~function~$\lambda (x).$
Namely, quartic integral has the~following form in our coordinates:
\begin{equation}
  F = \sum_{k=0}^4 \tilde{a}_k(x,y) p_x^{4-k} p_y^k,
\end{equation}
where~$\tilde{a}_k(x,y)$ are some unknown functions.
The~condition~$\{F,H\} = 0$ is equivalent to a~system of partial differential equations on the~functions~$\tilde{a}_k(x,y)$ and~$\lambda (x).$
Solving this system, we can obtain superintegrability condition as an~equation on~$\lambda (x)$, and an~explicit formula for the~quartic integral.

However this system of equations is extremely complicated, therefore it is useful to reduce our problem to a~simpler system of ordinary differential equations.
To do that, we need to eliminate one of the~variables.

Following~\cite{Matveev_Shevchishin}, let us consider a~linear operator
\begin{equation*}
  \L \colon F \mapsto \{ L, F \}
\end{equation*}
acting on the~space~$\J$ of complex quartic integrals in the~neighborhood of~$X$ (complexification is needed to guarantee that~$\L$ has eigenvalues).
$\L$ is well-defined: first, Poisson bracket of two integrals is an~integral due to the~Jacobi identity; secondly, $\deg \{L,F\} = \deg L + \deg F - 1 = \deg F.$
Note also that~$\J$ is a~finite-dimensional space.
Namely, $\dim \J \leq 15$~\cite{Kruglikov}.

It is readily seen that $L^4,\, L^2 H$ and~$H^2$ are eigenvectors of~$\L$ with an~eigenvalue of~$0$.
Two cases are possible:
\begin{enumerate}
\item
  There exists a~nonzero eigenvalue of~$\L$: $\mu \neq 0, \mu \in \C$.
\item
  The~spectrum of~$\L$ is~$\{0\}$.
\end{enumerate}

Let us consider the~basic case~$\mu \neq 0\,$; the~case~$\mu = 0$ will be considered later in Chapter~\ref{ch:mu_0}.

Assume there exists a~complex quartic integral
\begin{equation}
  F = \sum_{k=0}^4 \tilde{a}_k(x,y)\, p_x^{4-k} p_y^k
\end{equation}
(this time functions~$\tilde{a}_k(x,y)$ are complex-valued) such that $\{L, F\} = \d_y F = \mu F,$ i.e. $\d_y \tilde{a}_k(x,y) = \mu\, \tilde{a}_k(x,y),\, k = 0, \ldots, 4.$
It follows that $\tilde{a}_k(x,y) = e^{\mu y}\, a_k(x)$ for some functions~$a_k(x).$
Thus the~integral~$F$ has the~form
\begin{equation}
\label{eq:F}
  F = e^{\mu y} \sum_{k=0}^4 a_k(x)\, p_x^{4-k} p_y^k.
\end{equation}

For convenience in further calculations, let us rewrite the~metric as
\begin{equation}
  g = \frac{1}{h_x^2} (dx^2 + dy^2)
\end{equation}
for some function~$h(x).$
Then $H = \dfrac{h_x^2}{2} (p_x^2 + p_y^2)$, and the~condition~$\{F,H\} = 0$ reads
\begin{equation}
\label{eq:FH=0}
  \begin{aligned}
    e^{\mu y} h_x\, \left( a_0'(x) h_x - 4 a_0(x) h_{xx} \right) &\,\, p_x^5\\
    +\, e^{\mu y}\, h_x \left( \mu a_0(x) h_x + a_1'(x) h_x - 3 a_1(x) h_{xx} \right) &\,\, p_x^4 p_y\\
    +\, e^{\mu y}\, h_x \left( \mu a_1(x) h_x + a_2'(x) h_x - 4 a_0(x) h_{xx} - 2 a_2(x) h_{xx} \right) &\,\, p_x^3 p_y^2\\
    +\, e^{\mu y}\, h_x \left( \mu a_2(x) h_x + a_3'(x) h_x - 3 a_1(x) h_{xx} - a_3(x) h_{xx} \right) &\,\, p_x^2 p_y^3\\
    +\, e^{\mu y}\, h_x \left( \mu a_3(x) h_x + a_4'(x) h_x - 2 a_2(x) h_{xx} \right) &\,\, p_x p_y^4\\
    +\, e^{\mu y}\, h_x \left( \mu a_4(x) h_x - a_3(x) h_{xx} \right) &\,\, p_y^5 = 0.
  \end{aligned}
\end{equation}

Since the~monomials~$p_x^{5-k}p_y^k$ form a~basis in the~space of homogeneous polynomials of degree 5, every term in~\eqref{eq:FH=0} should vanish.
This gives us the~following system of 6 equations on the~functions~$h(x), a_0(x), \ldots, a_4(x)$:
\begin{subequations}
  \begin{align}
    \label{eq:sys_1}
    a_0'(x) h_x - 4 a_0(x) h_{xx} & = 0,\\
    \label{eq:sys_2}
    \mu a_0(x) h_x + a_1'(x) h_x - 3 a_1(x) h_{xx} & = 0,\\
    \label{eq:sys_3}
    \mu a_1(x) h_x + a_2'(x) h_x - 4 a_0(x) h_{xx} - 2 a_2(x) h_{xx} & = 0,\\
    \label{eq:sys_4}
    \mu a_2(x) h_x + a_3'(x) h_x - 3 a_1(x) h_{xx} - a_3(x) h_{xx} & = 0,\\
    \label{eq:sys_5}
    \mu a_3(x) h_x + a_4'(x) h_x - 2 a_2(x) h_{xx} & = 0,\\
    \label{eq:sys_6}
    \mu a_4(x) h_x - a_3(x) h_{xx} & = 0.
  \end{align}
\end{subequations}

Subsequently solving~\cref{eq:sys_1,eq:sys_2,eq:sys_3}, we get
\footnote{We choose such ``strange'' constants of integration to obtain explicit invariance of our equations under the~transformation~$\mu \to -\mu$, and to have ``correct'' limit~$\mu \to 0$.}
\begin{subequations}
  \begin{align}
    \label{eq:a0}
    a_0(x) & = \frac{A_0}{\mu} h_x^4,\\
    \label{eq:a1}
    a_1(x) & = (A_1 - A_0 h) h_x^3,\\
    \label{eq:a2}
    a_2(x) & =
    \left(
      \frac{1}{2} \mu A_0 h^2 + \frac{2A_0}{\mu} h_x^2 - \mu A_1 h + \frac{A_2}{\mu}
    \right) h_x^2,
  \end{align}
\end{subequations}
where~$A_0, A_1$ and~$A_2$ are complex constants of integration.
Also, we can eliminate~$a_4(x)$ using~\cref{eq:sys_6}:
\begin{equation}
  \label{eq:a4}
  a_4(x) = a_3(x) \cfrac{h_{xx}}{\mu\, h_x}.
\end{equation}

Unfortunately, there is no explicit way to eliminate~$a_3(x)$ after substituting these expressions for~$a_0(x),\, a_1(x),\, a_2(x)$ and~$a_4(x)$ into the~remaining equations~\cref{eq:sys_4,eq:sys_5}.
However we can use the~following trick.
Consider formal complex local coordinates
\begin{equation}
  z = x + \i y,\, \bz = x - \i y
\end{equation}
and the~corresponding local momenta
\begin{equation}
  p_z = \frac{1}{2} \left( p_x - \i p_y \right),\,
  \pbz = \frac{1}{2} \left( p_x + \i p_y \right).
\end{equation}
In the~new coordinates the~integral~$F$ reads
\begin{equation}
F = \sum_{k=0}^4 F_{4-k,k}(z, \bz)\, p_z^{4-k} \pbz^k.
\end{equation}
\begin{proposition}
  The~function~$F_{4,0}$ is holomorphic.
\end{proposition}
\begin{proof}
  In the~new coordinates we have~$H = 2\, h_x^2(z, \bz)\, p_z \pbz$.
  Poisson bracket~$\{ F, H \}$ can be expanded as
  \begin{equation*}
    \begin{gathered}
      \{ F, H \} = \left(
        \d_zF \d_{p_z}H - \d_{p_z}F \d_zH
      \right) + \left(
        \d_{\bz} F \d_{\pbz} H - \d_{\pbz} F \d_{\bz} H
      \right) = 0.
    \end{gathered}
  \end{equation*}
  Since the~bracket equals zero, the~coefficient of~$p_z^5$ should vanish.
  Clearly, this coefficient comes only from the~term~$\d_{\bz} F \d_{\pbz} H$.
  It is equal to~$\d_{\bz}F_{4,0}\, p_z^4 \cdot 2\,h_x^2\, p_z$, hence $\d_{\bz}F_{4,0} = 0$.
\end{proof}

Calculating~$F_{4,0}$ explicitly, we have
\begin{equation}
  \label{eq:F_40}
  \begin{gathered}
    F_{4,0} = e^{\mu y}
    \bigg[
      \left(
        \frac{h_{xx}}{\mu h_x} - \i
      \right) a_3(x) - \left(
        \frac{1}{2}\, \mu A_0 h^2 + \frac{A_0}{\mu} h_x^2
        - \mu A_1 h + \frac{A_2}{\mu} 
      \right) h_x^2\\
      +\, \i (A_1 - A_0 h) h_x^3
    \bigg]
    \equiv e^{\mu y}\,G(x).
  \end{gathered}
\end{equation}
Since~$F_{4,0} = e^{\mu y} G(x) = e^{-\i \mu z} \cdot e^{\i \mu x} G(x)$ is holomorphic, it follows that~$e^{\i \mu x} G(x)$ has to be constant:
\begin{equation}
  \label{eq:holom}
  e^{\i \mu x} G(x) = A_c,\quad G(x) = e^{-\i \mu x} A_c.
\end{equation}
Solving~\eqref{eq:F_40} for $a_3(x)$ and taking into account~\eqref{eq:holom}, we obtain
\begin{equation}
  \label{eq:a3}
  \begin{gathered}
    a_3(x) = \frac{1}{\dfrac{h_{xx}}{\mu h_x} - \i}\\ \times \Bigg[
      \frac{A_0}{\mu} h_x^4 + \i (A_0 h - A_1) h_x^3
      + \left(
        \frac{1}{2} \mu A_0 h^2 - \mu A_1 h + \frac{A_2}{\mu}
      \right) h_x^2 + e^{-\i \mu x} A_c
    \Bigg].
  \end{gathered}
\end{equation}

Finally, substituting the~obtained expressions for~$a_1(x), a_2(x), a_3(x)$ and~$a_4(x)$ into any equation~\cref{eq:sys_4,eq:sys_5}, we get an ordinary differential equation of order 3 on the~function~$h(x)$:
\begin{equation}
  \label{eq:3_order}
  \begin{gathered}
    \EEE(x,\, h,\, h_x,\, h_{xx},\, h_{xxx};\, \mu,\, A_0,\, A_1,\, A_2,\, A_c)\\
    \equiv \left[
      \frac{A_0}{\mu} h_x^4 + \i (A_0 h - A_1) h_x^3 +
      \left(
        \frac{1}{2} \mu A_0 h^2 - \mu A_1 h + \frac{A_2}{\mu}
      \right) h_x^2 + e^{-\i \mu x} A_c      
    \right]
    \cdot\, (h_{xxx} + \mu^2 h_x)\\
    -\, 3 \left[
      2\, \frac{A_0}{\mu} h_x^2 - \i (A_0 h - A_1) h_x +
      \left(
        \frac{1}{2} \mu A_0 h^2 - \mu A_1 h + \frac{A_2}{\mu}
      \right)      
    \right]
    \cdot \left(
      h_{xx} - \i \mu h_x
    \right) h_x h_{xx}\\
    - \frac{3}{\mu} (A_0 h - A_1) h_x h_{xx}^3\,= 0.
  \end{gathered}
\end{equation}

\subsection{Equation of order 2 in real and general cases}

Let us remark that in~\cite{Matveev_Shevchishin}, the~``$1+3$'' case was reduced to an~ordinary differential equation of order~3.
The~equation was integrated with an~integrating factor, and as a~result a~1st order ODE was obtained.
Therefore it is natural to~try to~generalize this result, and to~find an~integrating factor for the~``$1+4$'' case as~well.
Since equation~\eqref{eq:3_order} looks rather complicated, we are going to~obtain the~answer in a~special case.
Then we will be able to ``guess'' the~integrating factor in the~general case.

Namely, let us assume that~$\mu$ and~$A_k$ are real (we will show later that all real solutions of~\eqref{eq:3_order} can be obtained with real~$A_k$ and real or imaginary~$\mu$).
In this case, holomorphicity condition reads
\begin{equation}
  \label{eq:holomreal}
  e^{\i \mu x} G(x) = -\frac{A_4}{\mu} + \i A_3,
\end{equation}
where~$A_3,\, A_4$ are real.
This means that we have one more constant, thus we can expect to~obtain a~2nd order ODE on~$h(x)$.

Indeed, solving the~imaginary part of~\eqref{eq:holomreal} for~$a_3(x)$ and using~\eqref{eq:F_40}, we get
\begin{equation}
  \label{eq:a3real}
  a_3(x) = (A_1 - A_0 h) h_x^3 - \left(
    A_3 \cos (\mu x) + \frac{A_4}{\mu} \sin (\mu x)
  \right).
\end{equation}

Now, substituting this expression in the~real part of~\eqref{eq:holomreal}, we obtain a 2nd order~ODE on the~function~$h(x)$:
\begin{equation}
  \label{eq:2_order}
  \begin{gathered}
    \EE(x,\, h,\, h_x,\, h_{xx};\, \mu,\, A_0,\, A_1,\, A_2,\, A_3,\, A_4)\\
    \equiv \left[
      \left(
        A_0 h - A_1
      \right) h_x^3 +
      \left(
        A_3 \cos (\mu x) + \frac{A_4}{\mu} \sin(\mu x)
      \right)
    \right] h_{xx}\\
    + A_0 h_x^5
    + \left(
      \frac{1}{2}\, \mu^2 A_0 h^2 - \mu^2 A_1 h + A_2
    \right) h_x^3\\
    + \mu \left(
      A_3 \sin(\mu x) - \frac{A_4}{\mu} \cos(\mu x)
    \right) h_x = 0.
  \end{gathered}
\end{equation}

We shall now show that the~3rd order equation~\eqref{eq:3_order} follows from the~2nd order equation~\eqref{eq:2_order}.
Note that~\eqref{eq:2_order} depends on~$A_3$ and $A_4$, while~\eqref{eq:3_order} depends on~$A_c$.
Besides, $A_c = -\cfrac{A_4}{\mu} + \i A_3$.
Let us replace~$A_3$ and~$A_4$ by
\begin{equation}
  A_3 = \frac{1}{2 \i} (A_c - \bar{A}_c),\quad A_4 = -\frac{\mu}{2} (A_c + \bar{A}_c)
\end{equation}
in~\eqref{eq:2_order} (the~bar denotes complex conjugation here).
Solving the~resulting equation for~$\bar{A}_c$ and differentiating both sides, we obtain a~3rd order ODE proportional to~\eqref{eq:3_order}.
The~coefficient of proportionality equals
\begin{equation}
  \label{eq:mult_2}
  M^{(2)} = \frac{e^{-\i \mu x}\, h_x}{(h_{xx} - \i \mu\, h_x)^2}.
\end{equation}

It can easily be checked that this calculation can be ``reversed'' in the~general case as~well.
\begin{proposition}
  Order of equation~\eqref{eq:3_order} can be reduced by integration with integrating factor~\eqref{eq:mult_2}, constant of integration~$\bar{A}_c$ and a~linear change of constants
  \begin{equation}
    A_c = -\frac{A_4}{\mu} + \i A_3,\quad \bar{A}_c = -\frac{A_4}{\mu} - \i A_3.
  \end{equation}
  The resulting 2nd order equation is~\eqref{eq:2_order}.
\end{proposition}
\begin{remark}
  This time~$\bar{A}_c$ is a~new constant of integration independent of~$A_c\,;$ correspondingly, $A_3$ and~$A_4$ are arbitrary complex numbers like~$A_0,\, A_1,\, A_2$.
\end{remark}

Resolving~$h_{xx}$ from~\eqref{eq:2_order} and substituting it in~\eqref{eq:a3} and~\eqref{eq:a4}, we have
\begin{subequations}
  \begin{align}
    \label{eq:a3final}
    & a_3(x) = (A_1 - A_0 h) h_x^3 - \left(
      A_3 \cos (\mu x) + \frac{A_4}{\mu} \sin (\mu x)
    \right),\\
    \label{eq:a4final}
    & \begin{gathered}
      a_4(x) = \left(
        \frac{\mu}{2} (A_0 h - 2 A_1) h + \frac{A_0}{\mu} h_x^2 + \frac{A_2}{\mu}
      \right) h_x^2
      + \left(
        A_3 \sin (\mu x) - \frac{A_4}{\mu} \cos (\mu x)
      \right).
    \end{gathered}
  \end{align}
\end{subequations}
Therefore, we can now write the~integral~$F$ explicitly:
\begin{equation}
  \label{eq:F_explicit}
  \begin{gathered}
    F = e^{\mu y} \bigg\{
    \frac{A_0}{\mu}\, h_x^4\, p_x^4
    - \left(
      A_0 h - A_1
    \right) h_x^3\, p_x^3 p_y\\
    + \left[
      2\, \frac{A_0}{\mu}\, h_x^4
      + \left(
        \frac{\mu}{2} \left(
          A_0 h - 2 A_1
        \right) h + \frac{A_2}{\mu}
      \right) h_x^2
    \right] p_x^2 p_y^2\\
    - \left[
      \left(
        A_0 h - A_1
      \right) h_x^3
      + \left(
        A_3 \cos(\mu x) + \frac{A_4}{\mu} \sin(\mu x)
      \right)
    \right] p_x p_y^3\\
    + \bigg[
      \frac{A_0}{\mu}\, h_x^4
      + \left(
        \frac{\mu}{2} \left(
          A_0 h - 2 A_1
        \right) h + \frac{A_2}{\mu}
      \right) h_x^2\\
      + \left(
        A_3 \sin(\mu x) - \frac{A_4}{\mu} \cos(\mu x)
      \right)
    \bigg] p_y^4
    \bigg\}.
  \end{gathered}
\end{equation}

\subsection{Equation of order 1}

One can easily check the~following proposition:
\begin{proposition}
  Equation~\eqref{eq:2_order} reduces to the~equation
  \begin{equation}
    \label{eq:1_order_1}
    \begin{gathered}
      \Ef (x,\, h,\, h_x;\, \mu,\, A_0,\, A_1,\, A_2,\, A_3,\, A_4,\, A_5)\\
      \equiv 2 \left(
        A_0 h - A_1
      \right)
      \left(
        A_3 \cos(\mu x) + \frac{A_4}{\mu} \sin(\mu x)
      \right) h_x^3\\
      + \left[
        \left(
          \mu \left(
            A_0 h - 2 A_1
          \right) h + 2\, \frac{A_2}{\mu}
        \right)
        \left(
          A_3 \sin(\mu x) - \frac{A_4}{\mu} \cos(\mu x)
        \right) + A_5
      \right] h_x^2\\
      - \left(
        A_3 \cos(\mu x) + \frac{A_4}{\mu} \sin(\mu x)
      \right)^2 = 0
    \end{gathered}
  \end{equation}
  by integration with integrating factor
  \begin{equation}
    \label{eq:mult_11}
    M^{(1)}_1 = \frac{1}{h_x^3} \left( A_3 \cos(\mu x) + \frac{A_4}{\mu} \sin(\mu x) \right)
  \end{equation}
  and constant of integration~$A_5.$
\end{proposition}

Here is how one can ``guess'' the~equation~\eqref{eq:1_order_1} and the~integrating factor~\eqref{eq:mult_11}:
let us try to~express the~left-hand side of~\eqref{eq:2_order} in the~form
\begin{equation}
  \label{eq:ansatz}
  \EE = m(x) \left(
    \frac{d^2}{dx^2} + \mu^2
  \right) P(x),
\end{equation}
where~$m(x),\, P(x)$ are some functions (following the~case~``$1+3$''~\cite{Matveev_Shevchishin}).
Since~$A_k$ are independent constants, we can set all but one of them to be zeros.
This gives us the~following equations:
\begin{align}
  & m(x) = h_x^3,\\
  & P(x) = A_0 \mu h^2 - 2 A_1 h + 2 \frac{A_2}{\mu} + \left(
    A_3 \cos(\mu x) + \frac{A_4}{\mu} \sin(\mu x)
  \right) f(x).
\end{align}
Here the~function~$f(x)$ satisfies the~equation
\begin{equation}
  \begin{gathered}
    \frac{1}{2} \left(
      A_3 \cos(\mu x) + \frac{A_4}{\mu} \sin(\mu x)
    \right) h_x^3 f_{xx} + \mu \left(
      A_3 \sin(\mu x) - \frac{A_4}{\mu} \cos(\mu x)
    \right) h_x^3 f_x\\
    + \left(
      A_3 \cos(\mu x) + \frac{A_4}{\mu} \sin(\mu x)
    \right) h_{xx} + \mu \left(
      A_3 \sin(\mu x) - \frac{A_4}{\mu} \cos(\mu x)
    \right) h_x = 0,
  \end{gathered}
\end{equation}
which can be easily integrated:
\begin{equation}
  f_x = \frac{C}{\left(
      A_3 \cos(\mu x) + \frac{A_4}{\mu} \sin(\mu x)
    \right)^2} - \frac{1}{h_x^2}.
\end{equation}
Choosing $C = 0$, we finally obtain:
\begin{equation}
  \begin{gathered}
    P(x) = A_0 \mu h^2 - 2 A_1 h + 2 \frac{A_2}{\mu}
    - \left(
      A_3 \cos(\mu x) + \frac{A_4}{\mu} \sin(\mu x)
    \right) \int_0^x\frac{du}{h_x^2(u)}.
  \end{gathered}
\end{equation}
On the~other hand, from~\eqref{eq:ansatz} we get the~following:
\begin{equation}
  \begin{gathered}
    P(x) = B_1 \left(
      A_3 \cos(\mu x) + \frac{A_4}{\mu} \sin(\mu x)
    \right)
    +\, B_2 \left(
      -A_3 \sin(\mu x) + \frac{A_4}{\mu} \cos(\mu x)
    \right),
  \end{gathered}
\end{equation}
Differentiating the~resulting integro-differential equation, one can obtain equation~\eqref{eq:1_order_1} with~$A_5~=~\left( A_3^2 + \cfrac{A_4^2}{\mu^2} \right) B_1.$

\subsection{Poisson algebra of integrals}

Clearly, all the~equations on~$h(x)$ obtained so far are invariant under the~transformation~$\mu \to -\mu$.
Besides, the~functions~$a_k(x)$ transform as
$a_k(x) \to (-1)^k a_k(x).$
Therefore, if the~operator~$L$ has an~eigenvalue~$\mu$ and an~eigen integral
\begin{equation}
  \label{eq:Fplus}
  F_+ = e^{\mu y} \sum_{k=0}^4 a_k(x) p_x^{4-k} p_y^k,
\end{equation}
then it also has an~eigenvalue~$-\mu$ and eigen integral
\begin{equation}
  \label{eq:Fminus}
  F_- = e^{-\mu y} \sum_{k=0}^4 (-1)^k a_k(x) p_x^{4-k} p_y^k
\end{equation}
(actually, the~transformation~$\mu \to -\mu$ is equivalent to the~reverse of $y$ axis direction: $y \to -y,\, p_y \to -p_y$).

Let us now describe the~Poisson algebra generated by the~integrals~$H,\, L,\, F_+,\, F_-$.
Consider integrals~$G_7 = \{ F_+, F_- \}$ and~$G_8 = F_+ F_-$ of degrees 7 and 8.
\begin{proposition}
  \label{st:GLinvol}
  $\{ L, G_7 \} = \{ L, G_8 \} = 0.$
\end{proposition}
\begin{proof}
  We have:
  \begin{equation*}
    \begin{aligned}
      & \{ L, G_7 \} = \{ L, \{ F_+, F_- \} \} = 
      \{ \{ L, F_+ \}, F_- \} + \{ F_+, \{ L, F_- \} \}\\
      & = \mu \{ F_+, F_- \} -\mu \{ F_+, F_- \} = 0.
    \end{aligned}
  \end{equation*}
  Similarly,
  \begin{equation*}
    \{ L, G_8 \} = \{ L, F_+ F_- \} = 
    \{ L, F_+ \} F_- + F_+ \{ L, F_- \}
    = \mu F_+ F_- -\mu F_+ F_- = 0.
  \end{equation*}
\end{proof}

Since a~2-dimensional Hamiltonian system cannot have more than 2 functionally independent integrals in~involution, it follows from~\ref{st:GLinvol} that the~integrals~$G_7,\, G_8$ functionally (hence polynomially) depend on~$H$ and~$L$.
Let us assume that~$G_7$~($G_8$) is a linear combination of all possible integrals of degree 7 (8) constructed from~$L$ and~$H$, i.~e.
\begin{align}
  \label{eq:G7}
  G_7 & = C_{H^3L} H^3 L + C_{H^2L^3} H^2 L^3 + C_{HL^5} H L^5 + C_{L^7} L^7,\\
  \label{eq:G8}
  G_8 & = C_{H^4} H^4 + C_{H^3L^2} H^3 L^2 + C_{H^2L^4} H^2 L^4 + C_{HL^6} H L^6 + C_{L^8} L^8,
\end{align}
where~$C_\bullet$ are some constants.

Substituting the~explicit expressions for~$G_7$ (see~\eqref{eq:Fplus}, \eqref{eq:Fminus}, \cref{eq:a0,eq:a1,eq:a2}, \cref{eq:a3final,eq:a4final}), $L$ and~$H$ in~\eqref{eq:G7}, we obtain some linear combinations of 4 monomials~$p_x^{6-2k}p_y^{2k+1},\, k = 0,\ldots,3$.
It gives us 4 equations on the~constants~$C_\bullet.$
Solving these equations and taking into account~\eqref{eq:1_order_1}, we have
\begin{subequations}
  \begin{align}
    \label{eq:CH3L}
    C_{H^3L} & = \frac{4}{\mu}\, A_0 A_2 - 2 \mu A_1^2,\\
    \label{eq:CHL5}
    C_{HL^5} & = -6 \mu A_5,\\
    \label{eq:CL7}
    C_{L^7} & = 8 \left(
      \mu A_3^2 + \frac{A_4^2}{\mu}\,
    \right),
  \end{align}
  \begin{equation}
    \label{eq:1_order_2_initial}
    \begin{gathered}
      \left(
        A_0 h - A_1
      \right)^2 h_x^3 +
      \Bigg[
        \frac{\mu}{2} \left(
          A_0 h - 2 A_1
        \right) h
        \left(
          \frac{\mu}{2} \left(
            A_0 h - 2 A_1
          \right) h + 2\, \frac{A_2}{\mu}
        \right)\\
        + 2\, \frac{A_0}{\mu} \left(
          A_3 \sin(\mu x) - \frac{A_4}{\mu} \cos(\mu x)
        \right) + \frac{4A_2^2 - \mu\, C_{H^2L^3}}{4 \mu^2}
      \Bigg] h_x\\
      - 2 \left(
        A_0 h - A_1
      \right)
      \left(
        A_3 \cos(\mu x) + \frac{A_4}{\mu} \sin(\mu x)
      \right) = 0.
    \end{gathered}
  \end{equation}
\end{subequations}

Hence we obtain a~new equation of order~1 on the~function~$h(x)$.
Note that so far we are just assuming that it may be true (given that the~integral~$G_7$ has the~form~\eqref{eq:G7}).
However, denoting~$A_6 = 4A_2^2 - \mu C_{H^2L^3}$, resolving~$A_6$ from~\eqref{eq:1_order_2_initial} and differentiating both sides, we get the~2nd order equation~\eqref{eq:2_order} multiplied by
\begin{equation}
  \label{eq:mult_12}
  M_2^{(1)} = \frac{1}{h_x^2} \left(
    A_0 h - A_1
  \right).
\end{equation}

Therefore we have proven the~following
\begin{proposition}
  \label{st:eq_1_order_2}
  2nd order equation~\eqref{eq:2_order} can be reduced to the~equation
  \begin{equation}
    \label{eq:1_order_2}
    \begin{gathered}
      \Es(x,\, h,\, h_x;\, \mu,\, A_0,\, A_1,\, A_2,\, A_3,\, A_4,\, A_6)\\
      \equiv \left(
        A_0 h - A_1
      \right)^2 h_x^3 +
      \Bigg[
      \frac{\mu}{2} \left(
        A_0 h - 2 A_1
      \right) h
      \left(
        \frac{\mu}{2} \left(
          A_0 h - 2 A_1
        \right) h + 2\, \frac{A_2}{\mu}
      \right)\\
      + 2\, \frac{A_0}{\mu} \left(
        A_3 \sin(\mu x) - \frac{A_4}{\mu} \cos(\mu x)
      \right) + \frac{A_6}{4 \mu^2}
      \Bigg] h_x\\
      - 2 \left(
        A_0 h - A_1
      \right)
      \left(
        A_3 \cos(\mu x) + \frac{A_4}{\mu} \sin(\mu x)
      \right) = 0.
    \end{gathered}
  \end{equation}
  by integration with integrating factor~\eqref{eq:mult_12} and constant of integration~$A_6$.
\end{proposition}
\begin{Corollary}
  The~integral~$G_7$ has the~form~\eqref{eq:G7} with coefficients~\cref{eq:CH3L,eq:CHL5,eq:CL7} and
  \begin{equation}
    C_{H^2L^3} = \frac{1}{\mu}(4 A_2^2 - A_6).
  \end{equation}
\end{Corollary}

Now, taking into account both~\eqref{eq:1_order_1} and~\eqref{eq:1_order_2}, one can perform similar calculations for~$G_8$, and show that this integral indeed has the~form~\eqref{eq:G8} with coefficients
\begin{subequations}
  \begin{align}
    C_{H4} & = - \frac{A_0^2}{\mu^2},\\
    C_{H^3L^2} & = A_1^2 - 2 \frac{A_0A_2}{\mu^2},\\
    C_{H^2L^4} & = - \frac{1}{4 \mu^2} (4A_2^2 - A_6),\\
    C_{L^8} & = -\left(
      A_3^2 + \frac{A_4^2}{\mu^2}
    \right).
  \end{align}
\end{subequations}
Thus we can describe the~Poisson algebra of integrals as follows.
\begin{theorem}
  Poisson algebra generated by~$H,\, L,\, F_+$ and~$F_-$ has the~only algebraic relation
  \begin{equation}
    \begin{gathered}
      F_+ F_- = -\frac{A_0^2}{\mu^2}\, H^4
      + \left(
        A_1^2 - 2\, \frac{A_0 A_2}{\mu^2}
      \right) H^3 L^2\\
      - \frac{1}{4 \mu^2} \left(
        4 A_2^2 - A_6
      \right) H^2 L^4
      + A_5\, H L^6
      - \left(
        A_3^2 + \frac{A_4^2}{\mu^2}
      \right) L^8.
    \end{gathered}
  \end{equation}
  Poisson structure is defined by~$\{ H, \bullet \} = 0$, $\{ L, F_{\pm} \} = \pm \mu F_{\pm}$ and
  \begin{equation}
    \begin{gathered}
      \{ F_+, F_- \} = \left(
        \frac{4}{\mu}\, A_0 A_2 - 2 \mu A_1^2
      \right) H^3 L
      + \frac{1}{\mu} \left(
        4 A_2^2 - A_6
      \right) H^2 L^3\\
      - 6\, \mu\, A_5\, H L^5
      + 8 \left(
        \mu A_3^2 + \frac{A_4^2}{\mu}
      \right) L^7.
    \end{gathered}
  \end{equation}
\end{theorem}

\begin{remark}
  Algebraic interpretation of the~calculations performed in this chapter is the~following: denote by~$\D = \D \left(h(x),\, \cos(\mu x) \right)$ the~differential field generated by the~functions~$h(x)$ and~$\cos(\mu x)$.
  In~particular, this field contains~$\sin(\mu x)$ and all the~derivatives of~$h(x)$.
  1st order equation~\eqref{eq:1_order_1} generates the~ideal~$\I(\Ef)$ (all the~algebraic and differential consequences of the~equation).
  Obviously, $\EE \in \I(\Ef)$ and~$\EEE \in \I(\Ef)$.
  In this notation, solving equations on constants~$C_\bullet$ and taking into account~\eqref{eq:1_order_1} means the~corresponding formal calculations in the~field~$\D \big/ \I(\Ef)$.
  We obtain a~new nontrivial equation~\eqref{eq:1_order_2_initial}, i.e.~$\Es \notin \I(\Ef)$.
  Next, statement~\ref{st:eq_1_order_2} reads as~$\d\, \cfrac{\Es}{h_x} = \EE \in \I(\Ef)$.
  Finally, calculations for the~integral~$G_8$ taking into account both~\eqref{eq:1_order_1} and~\eqref{eq:1_order_2} means the~corresponding formal calculations in~$\D \big/ \I(\Ef,\, \Es)$.
\end{remark}

\subsection{Case~$\mu = 0$}
\label{ch:mu_0}

Let us now consider the~case when spectrum of~$\L$ is~$\{0\}$.
Note that the~linear span $\Span (L^4,\, L^2H,\, H^2)$ is an~eigenspace of~$\L$ with eigenvalue of~0.
Since 0 is the~only eigenvalue and dimension of~$\J$ is not less than~4 (there exists an~integral functionally independent with~$L$ and~$H$), it follows that there exists a~root integral~$F$ of height~2 such that $\L (F) \in \Span (L^4,\, L^2H,\, H^2)$), i.e.
\begin{equation}
  \label{eq:gen_eigen}
  \{ L, F \} = -A_4 L^4 + A_2 L^2H +A_0 H^2,
\end{equation}
where~$A_0,\, A_2,\, A_4$ are some constants\footnote{We choose the~names of constants in correspondence with the~general case as~$\mu \to 0$.}.

The~integral~$F$ has the~form
\begin{equation}
  F = \sum_{k=0}^4 \tilde{a}_k(x,y)\, p_x^{4-k} p_y^k.
\end{equation}
Then the~condition~\eqref{eq:gen_eigen} can be rewritten as
\begin{equation}
  \begin{gathered}
    \sum_{k=0}^4 \d_y\tilde{a}_k(x,y)\, p_x^{4-k} p_y^k\\
    = A_0 h_x^4\, p_x^4 + (2 A_0 h_x^2 + A_2) h_x^2\, p_x^2 p_y^2 +
    (A_0 h_x^4 + A_2 h_x^2 - A_4)\, p_y^4.
  \end{gathered}
\end{equation}
It follows that
\begin{subequations}
  \begin{align}
    \tilde{a}_0(x,y) & = a_0(x) + y\, A_0 h_x^4,\\
    \tilde{a}_1(x,y) & = a_1(x),\\
    \tilde{a}_2(x,y) & = a_2(x) + y\, (2 A_0 h_x^2 + A_2) h_x^2,\\
    \tilde{a}_3(x,y) & = a_3(x),\\
    \tilde{a}_4(x,y) & = a_4(x) + y\, (A_0 h_x^4 + A_2 h_x^2 - A_4).
  \end{align}
\end{subequations}
The~condition~$\{ F, H \} = 0$ reads
\begin{equation}
  \begin{aligned}
    h_x (a_0'(x) h_x - 4 a_0(x) h_{xx}) &\,\, p_x^5\\
    + h_x (a_1'(x) h_x - 3 a_1(x) h_{xx} + A_0 h_x^5) &\,\, p_x^4 p_y\\
    + h_x (a_2'(x) h_x - 2 a_2(x) h_{xx} - 4 a_0(x) h_{xx}) &\,\, p_x^3 p_y^2\\
    + h_x (a_3'(x) h_x - a_3(x) h_{xx} - 3 a_1(x) h_{xx} + 2 A_0 h_x^5 + A_2 h_x^3) &\,\, p_x^2 p_y^3\\
    + h_x (a_4'(x) h_x - 2 a_2(x) h_{xx}) &\,\, p_x p_y^4\\
    - h_x (a_3(x) h_{xx} - A_0 h_x^5 - A_2 h_x^3 + A_4 h_x) &\,\, p_y^5 = 0,
  \end{aligned}
\end{equation}
Once again, it gives us a~system of 6 equations:
\begin{subequations}
  \begin{align}
    \label{eq:sys0_1}
     a_0'(x) h_x - 4 a_0(x) h_{xx} & = 0,\\
    \label{eq:sys0_2}
     a_1'(x) h_x - 3 a_1(x) h_{xx} + A_0 h_x^5 & = 0,\\
    \label{eq:sys0_3}
     a_2'(x) h_x - 2 a_2(x) h_{xx} - 4 a_0(x) h_{xx} & = 0,\\
    \label{eq:sys0_4}
     a_3'(x) h_x - a_3(x) h_{xx} - 3 a_1(x) h_{xx} + 2 A_0 h_x^5 + A_2 h_x^3 & = 0,\\
    \label{eq:sys0_5}
     a_4'(x) h_x - 2 a_2(x) h_{xx} & = 0,\\
    \label{eq:sys0_6}
     a_3(x) h_{xx} - A_0 h_x^5 - A_2 h_x^3 + A_4 h_x & = 0.
  \end{align}
\end{subequations}

Subsequently solving~\cref{eq:sys0_1,eq:sys0_2,eq:sys0_3} and~\eqref{eq:sys0_5}, we find
\begin{subequations}
  \begin{align}
    \label{eq:a00}
    a_0(x) & = B_0 h_x^4,\\
    \label{eq:a10}
    a_1(x) & = (A_1 - A_0 h) h_x^3,\\
    \label{eq:a20}
    a_2(x) & = (2 B_0 h_x^2 + B_2) h_x^2,\\
    \label{eq:a40}
    a_4(x) & = B_0 h_x^4 + B_2 h_x^2 - B_4,
  \end{align}
\end{subequations}
where~$B_0,\, A_1,\, B_2,\, B_4$ are constants of integration.
Resolving~$a_3(x)$ from~\eqref{eq:sys0_6} and substituting the~obtained expressions for~$a_k(x)$ in~\eqref{eq:sys0_4}, we finally have the~3rd order ODE on~$h(x)$:
\begin{equation}
  \label{eq:3_order0}
  \begin{gathered}
    (A_0 h_x^4 + A_2 h_x^2 - A_4) h_{xxx} -
    3 (A_0 h - A_1) h_x h_{xx}^3
    - 3 (2 A_0 h_x^2 + A_2) h_x h_{xx}^2 = 0.
  \end{gathered}
\end{equation}

One can easily check that taking the~limit~$\mu \to 0$ in the~general case equation~\eqref{eq:3_order}, one obtains the~equation~\eqref{eq:3_order0}.
This gives us a~simple way to~integrate~\eqref{eq:3_order0} with integration factors~\eqref{eq:mult_2}, \eqref{eq:mult_11} and~\eqref{eq:mult_12} by taking the~limit~$\mu \to 0$.
Namely, integrating~\eqref{eq:3_order0} with the~factor
\begin{equation}
  M_0^{(2)} = \frac{h_x}{h_{xx}^2},
\end{equation}
and constant of integration~$A_3$, we have the~2nd order equation
\begin{equation}
  \label{eq:2_order0}
  \begin{gathered}
    \left(
      (A_0 h - A_1) h_x^3 + (A_3 + A_4 x)
    \right) h_{xx}
    + (A_0 h_x^4 + A_2 h_x^2 - A_4) h_x = 0,
  \end{gathered}
\end{equation}
which is the~limit of equation~\eqref{eq:2_order}.
Equation~\eqref{eq:2_order0} allows us to obtain the~final expression for~$a_3(x)$ without the~second derivative of~$h$:
combining~\eqref{eq:sys0_6} and~\eqref{eq:2_order0}, we find
\begin{equation}
  \label{eq:a30}
  a_3(x) = (A_1 - A_0 h) h_x^3 - (A_3 + A_4 x).
\end{equation}

Thus we have an~explicit formula for the~quartic integral:
\begin{equation}
  \label{eq:F_explicit0}
  \begin{gathered}
    F = -B_4 L^4 + B_2 L^2H +B_0 H^2
    +\, a_1(x)\, p_x^3 p_y + a_3(x)\, p_x p_y^3\\
    +\, y \left(
      A_0 h_x^4\, p_x^4 + (2 A_0 h_x^2 + A_2) h_x^2\, p_x^2 p_y^2 + (A_0 h_x^4 + A_2 h_x^2 - A_4) p_y^4
    \right).
  \end{gathered}
\end{equation}
Now we see that the~constants~$B_0,\, B_2,\, B_4$ are ``unnecessary'': in~\eqref{eq:F_explicit0} they correspond to the~term that is functionally dependent with~$L$ and~$H$.
This is why the~equations on~$h$ are independent of~$B_k$; we can set these constants to zero.
Note also that the~integral~\eqref{eq:F_explicit0} can be obtained from~$F_+$ and~$F_-$ as a~limiting case:
\begin{equation}
  F = \lim_{\mu \to 0} \frac{F_+ + F_-}{2}.
\end{equation}

Next, integrating~\eqref{eq:2_order0} with the~factor
\begin{equation}
  M_{01}^{(1)} = \frac{A_3 + A_4 x}{h_x^3}
\end{equation}
and constant of integration~$A_5$, we obtain the~1st order ODE
\begin{equation}
  \label{eq:1_order0_1}
  \begin{gathered}
    2 (A_0 h - A_1) (A_3 + A_4 x) h_x^3\\
    + (2 A_2 A_3 x - A_4 (A_0 h - 2 A_1) h + A_2 A_4 x^2 + A_5) h_x^2\\
    - (A_3 + A_4 x)^2 = 0,
  \end{gathered}
\end{equation}
which is the~limit of equation~\eqref{eq:1_order_1} with a~``fixed'' constant~$A_5 \to A_5 + 2\, \dfrac{A_2 A_4}{\mu^2}$.

Finally, integrating~\eqref{eq:2_order0} with the~factor
\begin{equation}
  M_{02}^{(1)} = \frac{A_0 h - A_1}{h_x^2}
\end{equation}
and constant of integration~$A_6$, we find the~second equation of order 1:
\begin{equation}
  \label{eq:1_order0_2}
  \begin{gathered}
    (A_0 h - A_1)^2 h_x^3\\
    + \left(
      A_2 (A_0 h - 2 A_1) h + 2 A_0 A_3 x + A_0 A_4 x^2 + A_6
    \right) h_x\\
    - 2 (A_0 h - A_1) (A_3 + A_4 x) = 0,
  \end{gathered}
\end{equation}
which is the~limit of equation~\eqref{eq:1_order_2} with a~``fixed'' constant~$A_6 \to 4 \mu^2 A_6 + 8 A_0 A_4$.

\subsection{Real solutions}
\label{ch:real}

So~far, we have assumed that~$\mu$ and~$A_k$ are complex, therefore solutions of the~system~\cref{eq:1_order_1,eq:1_order_2} are complex as well.
However, we are only interested in real solutions, so it is necessary to~find the~corresponding conditions on the~parameters.

Similarly to the~case~``$1+3$''~\cite{Matveev_Shevchishin}, it can be shown that if a~function~$h(x)$ satisfies two systems of the~form~\cref{eq:1_order_1,eq:1_order_2} with some complex parameters~$(\mu,\, A_k)$ and~$(\mu',\, A_k')$, then
\begin{equation}
  \begin{gathered}
  \mu' = \pm \mu,\\
  (A_0',\, A_1',\, A_2',\, A_3',\, A_4',\, A_5',\, A_6')
  = (\lambda A_0,\, \lambda A_1,\, \lambda A_2,\, \lambda A_3,\,
  \lambda A_4,\, \lambda^2 A_5,\, \lambda^2 A_6).
  \end{gathered}
\end{equation}

Now let~$h(x)$ be a~real-valued solution of~\cref{eq:1_order_1,eq:1_order_2} with some constants~$(\mu,\, A_k)$.
Then~$h(x)$ also satisfies the~complex conjugate system with parameters~$(\bar{\mu},\, \bar{A}_k)$.
It follows that $\bar{\mu} = \pm \mu$, so $\mu$ is real or purely imaginary.
Secondly, $\bar{A}_k = \lambda^{n(k)} A_k$, where~$n(k) = 1$ for~$k = 0,\ldots,4$ and~$n(k) = 2$ for~$k = 5,\, 6$.
Obviously, $\left| \lambda \right| = 1,\, \lambda = e^{\i \phi}$, and~$h(x)$ also satisfies the~system with parameters~$(\mu,\, e^{\frac{1}{2}\i n(k) \phi} A_k)$; here~$e^{\frac{1}{2}\i n(k) \phi} A_k$ are real.
Thus we can assume that~$A_k$ are real, and~$\mu$~is real or purely imaginary.

When~$\mu$ is real, all the~formulas obtained earlier (equations, expressions for~$F_\pm,\, G_7,\, G_8$) become real.
When~$\mu$ is imaginary, one can substitute~$\mu \to \i \mu$.
Trigonometric functions of imaginary argument become hyperbolic functions of real argument:
$\cos(\mu x) \to \cosh(\mu x),\, \cfrac{\sin(\mu x)}{\mu} \to \cfrac{\sinh(\mu x)}{\mu}$.
Instead of complex integrals~$F_+,\, F_-$ one should take real linear combinations
\begin{equation}
  F_1 = \frac{1}{2} (F_+ + F_-),\quad F_2 = \frac{1}{2 \i} (F_+ - F_-).
\end{equation}
Integral~$G_8$ turns out to be real:
\begin{equation}
  \begin{gathered}
    G_8 = F_+ F_- = (F_1 + \i F_2) (F_1 - \i F_2) = F_1^2 + F_2^2,
  \end{gathered}
\end{equation}
while~$G_7$ is imaginary:
\begin{equation}
  \begin{gathered}
    G_7 = \{ F_+, F_- \} = \{ F_1 + \i F_2, F_1 - \i F_2 \} =
    -2 \i \{F_1, F_2\},
  \end{gathered}
\end{equation}
so instead of~$G_7$ one can take~$\i G_7.$

These results can be summarized as follows.
We have 3 different cases: nonzero real~$\mu$ (trigonometric case), nonzero imaginary~$\mu$ (hyperbolic case) and special case~$\mu = 0$ (linear case).
Considering these cases, we obtain formulas~$a,\, b$ and~$c$ from Theorem~\ref{th:local} respectively.
This concludes the~proof. 

\begin{remark}
  \label{rem:A0}
  It is easy to~show that in the~case~$A_0 = 0$ integrals~$F_+$ and~$F_-$ are proportional to~$p_y = L$, i.e. they are products of linear and cubic integrals.
  On the~other hand, if we set $A_0 = 0$ in~\eqref{eq:main}, then the~second equation coincides with the~equation obtained in~\cite{Matveev_Shevchishin} for the~case~``$1+3$'' (up to the~constant names).
\end{remark}

\section{Further study of the~metric equations}
\label{ch:further}

\subsection{Symmetries and normal forms of equations}

Local description of ``$1+4$'' superintegrable metrics obtained in the~previous chapter is rather redundant.
Non-unique choice of several objects (local coordinate system, function~$h(x)$, quartic integral~$F$) results in a~large symmetry group.
It allows to~simplify the~equations describing the~metric.
Consider the~following symmetries.

\begin{enumerate}
\item {\bf Coordinate changes}

  Local coordinate system such that
  \begin{equation}
    \label{eq:g}
    g = \frac{1}{h_x^2} (dx^2 + dy^2)
  \end{equation}
  is non-unique.
  Clearly, the~following 5 transformations leave it unchanged:
  \begin{gather*}
    x \to x + a;\\
    y \to y + a;\\
    x \to -x;\\
    y \to -y;\\
    x \to \lambda x,\, y \to \lambda y.
  \end{gather*}

\item {\bf ``Gauge'' symmetry}

  The~metric~$g$ depends on the~derivative of~$h$, hence this function can be translated as
  \begin{equation*}
    h \to h + a
  \end{equation*}
  without any change in the~metric.

\item {\bf Homogeneity}

  In Chapter~\ref{ch:real} we already noted that the~metric equations are invariant under multiplication of~$A_k$ by an~arbitrary nonzero~$\lambda$ to the~necessary power.
  This is because quartic integral~$F$ is defined up to an~arbitrary nonzero constant~$\lambda$:
  \begin{equation*}
    F \to \lambda F.
  \end{equation*}

\item {\bf Metric dilations}

  Transformation
  \begin{equation*}
    h \to \lambda h,
  \end{equation*}
  changes the~metric.
  Nevertheless, it can be considered a~symmetry because it simply rescales the~problem, hence it maps solutions to solutions and does not affect superintegrability.
\end{enumerate}

This symmetry group acts naturally on the~space of parameters~$(\mu,\, A_0,\, A_1,\ldots, A_6)$.
Orbits of this action are sets of parameters which correspond to the~same metrics up to dilations.
Therefore it is sufficient to~consider only the~``simplest'' representative of each orbit.
Let us call the~equations corresponding to these representatives {\it normal forms}.

Now let us describe normal forms for each of the~3 cases:

\begin{enumerate}
\item {\bf Trigonometric case}

  Rescaling~$x$ and changing its sign if necessary we can set~$\mu = 1$.
  Next, translating~$x$ we can shift the~phase inside trigonometric functions so that~$A_4 = 0$.
  Homogeneity allows us to~set~$A_3 = 1$.
  Besides that, gauge symmetry~$h \to h + a$ makes it possible to~set~$A_1 = 0$ (here we assume that~$A_0 \neq 0$; see~\ref{rem:A0}).
  Finally, it is easy to~check that dilating the~metric we can obtain~$A_0 = \pm 1$ depending on the~initial sign of~$A_0$.
  It follows that in this case we have two normal forms:
  \begin{equation}
    \label{eq:trig_normal}
    \begin{cases}
      \begin{gathered}
        \pm\, 2 \cos(x)\, h h_x^3
        + \left(
          (\pm h^2 + 2 A_2) \sin(x) + A_5
        \right) h_x^2\,
        - \cos^2(x) = 0,
      \end{gathered}\\
      \begin{gathered}
        h^2 h_x^3 + \left[
          \frac{h^4}{4} \pm A_2 h^2 \pm 2 \sin(x) + \frac{A_6}{4}
        \right] h_x \mp 2 h \cos(x) = 0,
      \end{gathered}
    \end{cases}
  \end{equation}
  depending on 3 parameters~$A_2,\, A_5$ and~$A_6$.

\item {\bf Hyperbolic case}

  The~main difference from the~trigonometric case is that we cannot set~$A_4$ to zero by translation of~$x$.
  Depending on the~ratio of~$A_3$ and~$A_4$ linear combination~$A_3 \cosh(x) + A_4 \sinh(x)$ can be transformed to one of three functions:~$A_3 \cosh(x),\, A_4 \sinh(x)$ or~$A_3\, e^x$.
  Therefore, we have 6 normal forms:
  \begin{equation}
    \label{eq:hyper_normal}
    \begin{aligned}
      & \begin{cases}
        \begin{gathered}
          \pm 2 \cosh(x) h h_x^3 - \left(
            (\pm h^2 - 2 A_2) \sinh(x) - A_5
          \right) h_x^2 - \cosh^2(x) = 0,
        \end{gathered}\\
        \begin{gathered}
          h^2 h_x^3 - \left(
            \frac{h^4}{4} \mp A_2 h^2 \mp 2 \sinh(x) + \frac{A_6}{4}
          \right) h_x \mp 2 \cosh(x) h = 0,
        \end{gathered}
      \end{cases}\\
      & \begin{cases}
        \begin{gathered}
          \pm 2 \sinh(x) h h_x^3 - \left(
            (\pm h^2 - 2 A_2) \cosh(x) - A_5
          \right) h_x^2 - \sinh^2(x) = 0,
        \end{gathered}\\
        \begin{gathered}
          h^2 h_x^3 - \left(
            \frac{h^4}{4} \mp A_2 h^2 \mp 2 \cosh(x) + \frac{A_6}{4}
          \right) h_x \mp 2 \sinh(x) h = 0,
        \end{gathered}
      \end{cases}\\
      & \begin{cases}
        \begin{gathered}
          \pm 2\, e^x h h_x^3 - \left(
            (\pm h^2 - 2 A_2) e^x - A_5
          \right) h_x^2 - e^{2x} = 0,
        \end{gathered}\\
        \begin{gathered}
          h^2 h_x^3 - \left(
            \frac{h^4}{4} \mp A_2 h^2 \mp 2 e^x + \frac{A_6}{4}
          \right) h_x \mp 2 e^x h = 0
        \end{gathered}
      \end{cases}
    \end{aligned}
  \end{equation}
  with parameters~$A_2,\, A_5$ and~$A_6$.

\item {\bf Linear case}

  Arguing as above, we get the~following normal forms:
  \begin{equation}
    \label{eq:linear_normal}
    \begin{cases}
      \begin{gathered}
        \pm 2\, x\, h h_x^3 + \left(
          \mp h^2 + A_2 x^2 + A_5
        \right) h_x^2 - x^2 = 0,
      \end{gathered}\\
      \begin{gathered}
        h^2 h_x^3 + \left(
          \pm A_2 h^2 \pm x^2 + A_6
        \right) h_x \mp 2\, x h = 0
      \end{gathered}
    \end{cases}
  \end{equation}
  with parameters~$A_2,\, A_5$ and~$A_6$.
\end{enumerate}

\subsection{Algebraic form of equations}
\label{ch:algebraic}

In the~linear case~\eqref{eq:linear_normal} system of~equations on~$h,\, h_x$ is algebraic, so it defines an~algebraic curve in the~space~$(x,\,h,\,h_x)$.
Let us show that in the~remaining cases solutions can be described as algebraic curves as well, if we choose appropriate coordinates.

Actually it follows from the fact that a~circle and a~hyperbola are rational curves.
Therefore pairs of functions~$(\cos(x), \sin(x))$ and~$(\cosh(x), \sinh(x))$ can be transformed to rational functions by a~coordinate change.
Namely, in the~trigonometric case we locally set
\begin{equation}
  r = \tan(\frac{x}{2}),\quad
  \cos(x) = \frac{1-r^2}{1+r^2},\quad
  \sin(x) = \frac{2r}{1+r^2}.
\end{equation}
In the~hyperbolic case we set (globally)
\begin{equation}
  r = e^x,\quad
  \cosh(x) = \frac{r^2+1}{2r},\quad
  \sinh(x) = \frac{r^2-1}{2r}.
\end{equation}
Next, let~$p = h_x$ (actually $h_r$ depends rationally on~$h_x$, but we work with the~$x$-derivative for convenience).
Then, in the~coordinates~$(r,\, h,\, p)$ normal forms can be rewritten as follows:
\begin{equation}
  \begin{cases}
    \begin{gathered}
      \pm 2 (1 - r^4) h p^3 - (1 - r^2)^2\\
      + \left(
        \pm 2 r h^2 + 4 A_2 r + A_5 (1 + r^2)
      \right) (1 + r^2) p^2 = 0,
    \end{gathered}\\
    \begin{gathered}
      4 (1 + r^2) h p^3 \mp 8 (1 - r^2) h\\
      + \left(
        (1 + r^2) h^4 \pm 4 A_2 (1 + r^2) h^2 + A_6 (1 + r^2) \pm 16 r
      \right) p = 0
    \end{gathered}
  \end{cases}
\end{equation}
in the~trigonometric case and
\begin{subequations}
  \begin{align}
    & \begin{cases}
      \label{eq:form_global}
      \begin{gathered}
        \pm 4 r (r^2 + 1) h p^3 - (r^2 + 1)^2\\
        + \left(
          \mp 2 r (r^2 - 1) h^2 + 4 A_2 r (r^2 - 1) + 4 A_5 r^2
        \right) p^2  = 0,
      \end{gathered}\\
      \begin{gathered}
        4 r h^2 p^3 \mp 4 (r^2 + 1) h\\
        + \left(
          - r h^4 \pm 4 A_2 r h^2 \pm 4 r^2 - A_6 r - 4
        \right) p = 0,
      \end{gathered}
    \end{cases}\\
    & \begin{cases}
      \begin{gathered}
        \pm 4 r (r^2 - 1) h p^3 - (r^2 - 1)^2\\
        + \left(
          \mp 2 r (r^2 + 1) h^2 + 4 A_2 r (r^2 + 1) + 4 A_5 r^2
        \right) p^2  = 0,
      \end{gathered}\\
      \begin{gathered}
        4 r h^2 p^3 \mp 4 (r^2 - 1) h\\
        + \left(
          - r h^4 \pm 4 A_2 r h^2 \pm 4 r^2 - A_6 r + 4
        \right) p = 0,
      \end{gathered}
    \end{cases}\\
    & \begin{cases}
      \begin{gathered}
        \pm 2 r h p^3 - r^2 + \left(
          \mp r h^2 + 2 A_2 r + A_5
        \right) p^2  = 0,
      \end{gathered}\\
      \begin{gathered}
        4 h^2 p^3 \mp 8 r h + \left(
          - h^4 \pm 4 A_2 h^2 \pm 8 r - A_6
        \right) p = 0
      \end{gathered}
    \end{cases}
  \end{align}
\end{subequations}
in the~hyperbolic case.

\subsection{Existence of metrics globally defined on a~sphere}

Suppose our~local coordinate chart~$(x, y)$ is an~almost global chart on a~sphere, so that it covers the~whole sphere except for the~poles.
Suppose also that different values of~$x$ correspond to different~latitudes, and poles of a~sphere are limiting points as~$x \to \pm \infty$.
Similarly, let different values of~$y$ correspond to different longitudes.
If a~solution of~\eqref{eq:main} is defined on the~whole real axis, and the~corresponding metric extends smoothly to the~poles (together with the~integrals), then such superintegrable system is defined globally on a~sphere.

The~metric is well-defined if~$h_x = p \neq 0$ for all~$x \in \R$.
Clearly, this condition holds for the~equations with algebraic form~\eqref{eq:form_global}.
Next, the~metric is defined on the~whole real axis (and is non-degenerate) if~$h_x$ does not become infinity at any point.

\begin{proposition}
  There exists an~uncountable set of parameter tuples~$(A_2,\, A_5,\, A_6)$ such that the~system~\eqref{eq:form_global} with the~lower sign has no points~$R \in (0, \infty)$ such that~$\lim\limits_{r \to R} p(r) = \infty$.
\end{proposition}
\begin{proof}
  Let~$R \in (0, \infty)$ be a~point such that~$\lim\limits_{r \to R} p(r) = \infty$.
  Then the~leading coefficient (as a~polynomial in~$p$) should vanish at the~point~$R$ for both equations~\eqref{eq:form_global}, hence~$h(R) = 0$.
  Let us rewrite our system of equations (with the~lower sign) as follows:
  \begin{equation}
    \begin{cases}
      \begin{gathered}
        \left(
          -4 r (r^2 + 1) h p - 2 r (r^2 - 1) h^2 + 4 A_2 r (r^2 - 1) + 4 A_5 r^2
        \right) p^2
        - (r^2 + 1)^2 = 0,
      \end{gathered}\\
      \begin{gathered}
        \left(
          4 r h^2 p^2 - r h^4 - 4 A_2 r h^2 - 4 r^2 - A_6 r - 4
        \right) p
        + 4 (r^2 + 1) h = 0.
      \end{gathered}
    \end{cases}
  \end{equation}
  Therefore, expressions
  \begin{gather*}
    -4 r (r^2 + 1) h p - 2 r (r^2 - 1) h^2 + 4 A_2 r (r^2 - 1) + 4 A_5 r^2,\\
    4 r h^2 p^2 - r h^4 - 4 A_2 r h^2 - 4 r^2 - A_6 r - 4
  \end{gather*}
  vanish at~$R$, so
  \begin{gather*}
    \lim_{r \to R} h(r) p(r) = \frac{A_2 (R^2 - 1) + A_5 R}{R^2 + 1},\\
    \lim_{r \to R} h(r)^2 p(r)^2 = \frac{R^2 + \frac{A_6}{4} R + 1}{R}.
  \end{gather*}
  Comparing these expressions, we find:
  \begin{gather*}
    \left(
      \frac{A_2 (R^2 - 1) + A_5 R}{R^2 + 1}
    \right)^2 = \frac{R^2 + \dfrac{A_6}{4} R + 1}{R},\\
    \left(
      A_2 (R^2 - 1) + A_5 R
    \right)^2 R = \left(
      R^2 + \frac{A_6}{4} R + 1
    \right) (R^2 + 1)^2.
  \end{gather*}
  For~$A_2 \neq 0,\, A_6 \geq 0$ plots of left-hand side (LHS) and right-hand side (RHS) of the~last equation can be qualitatively drawn as in~Fig.~\ref{fig:global} (LHS is drawn in red, RHS is drawn in blue).
  \begin{figure}[t]
    \centering
    \includegraphics[width=0.3\textwidth]{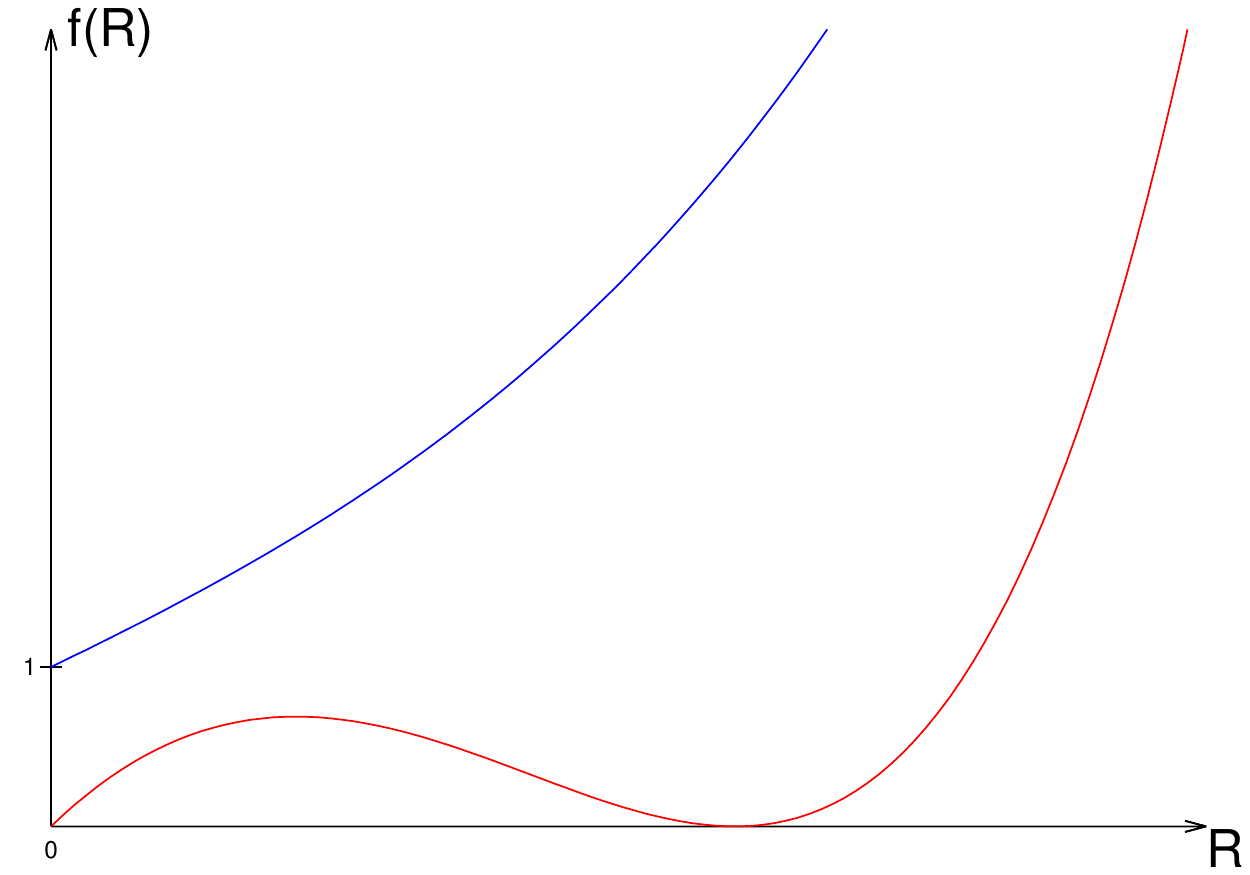}
    \caption{LHS vs. RHS}
    \label{fig:global}
  \end{figure}
  Besides, LHS grows as~$R^5$, while RHS grows as~$R^6$.
  To conclude the~proof, it remains to~note that the~plots do not intersect for uncountably many sets of parameters~$(A_2,\, A_5,\, A_6)$ (e.g., sufficiently small~$A_2$ and~$A_5$ with fixed~$A_6$), therefore~$p(r)$ is finite at any point.
\end{proof}

So, the~solution is defined on the~whole real axis for parameter values mentioned above.
We now show that the~metric extends smoothly to the~poles.
Since the~solution lies on an~algebraic curve, it follows that~$h(r)$ and~$p(r)$ can be expanded as Puiseux series in a~neighborhood of~$r = 0$:
\begin{equation}
  \label{eq:puiseux}
  \begin{gathered}
    h(r) = r^{m/n} \sum_{k = 0}^{\infty} a_k r^{k/n} = r^{m/n} f(r^{1/n}),\\
    p(r) = r h_r = r^{m/n} \sum_{k = 0}^{\infty} \frac{m+k}{n} a_k r^{k/n} = r^{m/n} f_1(r^{1/n}),
  \end{gathered}
\end{equation}
where~$m \in \Z,\, n \in \N$, $f$ and~$f_1$ are real analytic functions, and~$f(0) \neq 0$.
Substituting~\eqref{eq:puiseux} for~$h(r)$ and~$p(r)$ in~\eqref{eq:form_global} one can easily see that the~expansions start with the~power~$- \frac{1}{2}$ or~$- \frac{1}{4}$, i.e.
\begin{equation}
  \begin{gathered}
    h(r) = r^{-1/n} f(r^{1/n}),\\
    p(r) = r^{-1/n} f_1(r^{1/n})
  \end{gathered}
\end{equation}
with~$n = 2$ or~$4$, and $f_1(0) \neq 0$.

Define new coordinates~$(\rho, \phi)$ by the~formulas~$\rho = r^{1/n} = e^{x/n},\, \phi = \dfrac{y}{n}$.
In coordinates~$(\rho, \phi)$ metric has the~form
\begin{equation}
  \begin{gathered}
    g = \frac{1}{h_x^2} (dx^2 + dy^2) =
    \frac{1}{p^2} (n^2\, \frac{d\rho^2}{\rho^2} + n^2 d\phi^2)
    = \frac{n^2}{f_1(\rho)^2} (d\rho^2 + \rho^2 d\phi^2).
  \end{gathered}
\end{equation}

Since~$f_1$ is real analytic and non-vanishing in a~neighborhood of the~origin~$\rho = 0$, it follows that the~metric~$g$ behaves as an~ordinary Euclidean metric in polar coordinates, therefore it extends smoothly to the~south pole.
By changing the~sign of~$x$ (i.e. switching the~poles), it can be proved that the~metric extends to the~north pole as well.

To~show the~extensibility of the~integrals, consider the~following coordinates:
\begin{equation}
  \xi = r \cos(y),\quad \eta = r \sin(y)
\end{equation}
and the~corresponding momenta~$\pxi,\, \peta$.
South pole is the~origin~$(0,\, 0)$ in new coordinates.
The~momenta~$p_x,\, p_y$ are given by
\begin{equation}
  \begin{gathered}
    p_x = \xi\, \pxi + \eta\, \peta,\\
    p_y = \xi\, \peta - \eta\, \pxi,
  \end{gathered}
\end{equation}
It follows that~$L = p_y$ extends smoothly to the~origin.
By the~pole switching argument $L$ extends also to the~north pole.

Quartic integral~$F_1$ has the~following form in new coordinates:
\begin{equation*}
  \begin{gathered}
    F_1 = 
    \left[
      \xi \eta^4 + \frac{1}{2} \xi r h_x^2 \Big(
        \eta^2 (h^2 + 2 h h_x) - 2 r^2 h_x^2 - 2A_2 \eta^2
      \Big)
    \right] \pxi^4\\
    + \left[
      \frac{\eta^3}{2} (1 + \eta^2 - 7 \xi^2) - \eta r h_x^2 \left(
        \xi^2 (h^2 + h h_x) - \eta^2 h h_x - 2 \xi^2 A_2
      \right)
    \right] \pxi^3 \peta\\
    + \left[
      -\frac{3}{2} \xi \eta^2 (1 + \eta^2 - 3 \xi^2) - \frac{1}{2} \xi r^3 h_x^2 \left(
        4 h_x^2 - h^2 + 2 A_2
      \right)
    \right] \pxi^2 \peta^2\\
    + \left[
      \frac{1}{2} \xi^2 \eta (3 + 3 \eta^2 - 5 \xi^2) - \eta r h_x^2 \left(
        \xi^2 (h^2 + h h_x) - \eta^2 h h_x - 2 \xi^2 A_2
      \right)
    \right] \pxi \peta^3\\
    + \left[
      -\frac{\xi^3}{2} (1 + \eta^2 - \xi^2) - \frac{1}{2} \xi r h_x^2 \left(
        2 \eta^2 h h_x - \xi^2 h^2 + 2 r^2 h_x^2 + 2 A_2 \xi^2
      \right)
    \right] \peta^4.
  \end{gathered}
\end{equation*}

Since functions~$h(r),\, h_x(r)$ are of order~$r^{-1/2}$ as~$r \to 0$, it is sufficient to check that
\begin{equation*}
  \begin{gathered}
    h^2 + 2 h h_x,\quad
    \xi^2 (h^2 + h h_x) - \eta^2 h h_x,\quad
    \xi^2 h^2 - 2 \eta^2 h h_x
  \end{gathered}
\end{equation*}
have finite limits as~$r \to 0$.
It follows from the~expansions
\begin{equation*}
  \begin{gathered}
    h^2 = c_0^2 r^{-1} + \mbox{r.p.},\\
    h h_x = -\frac{1}{2} c_0^2 r^{-1} + \mbox{r.p.}
  \end{gathered}
\end{equation*}
(r.p. stands for the~regular part of a~series), and identity~$r^2 = \xi^2 + \eta^2$.
We see that~$F_1$ extends smoothly to the~south pole.
Once again, extensibility to the~north pole follows from the~pole switching argument.

Continuing this line of reasoning, one can easily find that~$F_2$ also extends smoothly to the~poles.
This completes the~construction of metrics globally defined on a~sphere~$S^2$.

\newpage


\bibliographystyle{utphys}
\bibliography{superintegrable}

\end{document}